\documentclass[aps,pra,amsmath,amssymb,amsfonts,superscriptaddress,floatfix,nofootinbib]{revtex4}
\pdfoutput=1
\usepackage{algorithm}
\usepackage{amsmath,amsthm,amsfonts,amssymb}
\usepackage{graphicx}

\usepackage{color}
\usepackage{hyperref}

\newtheorem{theorem}{Theorem}

\newtheorem{lemma}{Lemma}

\usepackage[capitalize,nameinlink]{cleveref}

\newcommand{\be}{\begin{equation}}
\newcommand{\ee}{\end{equation}}

\newcommand{\psiinputun}{\Psi_{\rm input}^{\rm unnormalized}}

\newcommand{\psiideal}{\Psi_{\rm ideal}}
\newcommand{\expec}{\mathbb{E}}
\newcommand{\nboseq}{n_{bos}^{eq}}

\newcommand{\vsig}{v_{\rm sig}}
\newcommand{\Tspike}{T_0}
\newcommand{\Tplus}{T^+}
\newcommand{\Tminus}{T^-}
\newcommand{\nbos}{{n_{bos}}}

\newcommand{\Pis}{\Pi_{\rm symm}}
\newcommand{\rhosymmmm}{\rho_{\rm symm,mm}}

\newcommand{\emax}{E_{max}}
\newcommand{\emaxfuture}{E_{max,future}}
\newcommand{\dem}{E'}
\newcommand{\normmin}{P_{>}}
\newcommand{\ecut}{E_{cut}}
\newcommand{\Psisig}{\Psi_{\rm sig}}
\newcommand{\psiinput}{\Psi_{\rm input}}

\newcommand{\prob}{{\rm Pr}}

\begin{document}

\title{Accelerating Classical and Quantum Tensor PCA}
\author{Matthew B.~Hastings}
\begin{abstract}
Spectral methods are a leading approach for tensor PCA with a ``spiked" Gaussian tensor\cite{wein2019kikuchi,hastings2020classical}.  The methods use the spectrum of a linear operator in a vector space with exponentially high dimension and in Ref.~\cite{hastings2020classical} it was shown that quantum algorithms could then lead to an exponential space saving (in a fairly standard way) as well as a \emph{quartic} speedup over classical.  Here we show how to accelerate both classical and quantum algorithms quadratically, while maintaining the same quartic separation between them.  That is, our classical algorithm here is quadratically faster than the original classical algorithm, while the quantum algorithm is eigth-power faster than the original classical algorithm.  We then give a further modification of the quantum algorithm, increasing its speedup over the modified classical algorithm to the sixth power, so that it is twelth-power faster than the original classical algorithm.
We only prove these speedups for detection, rather than recovery, but we give a strong plausibility argument that our algorithm achieves recovery also.  

Note added: After this paper was prepared, A. Schmidhuber pointed out to me Ref.~\cite{kothari2026smooth}.  This reference improves the best existing bounds on the spectral norm of a certain random operator.  Because the norm of this operator enters into the runtime of the previous algorithms, with this improvement on the norm, we no longer have a provable polynomial speedup.  Our results are phrased in terms of certain properties of the spectrum of this operator (not merely the largest eigenvalue but also the density of states).  So, if these properties still hold, the speedup still holds.  However, due to the improvement in the norm bound, we no longer know that the density of states bound that we need holds compared to the best norm bound.  Rather than modify the paper, I have left it unchanged but added a section at the end discussing the needed property of density of states and considering for which problems (there are several problems for which this kind of quartic quantum speedup has been used and the techniques here will likely be applicable to several of them) the property is likely to hold.
\end{abstract}
\maketitle

In \cite{wein2019kikuchi,hastings2020classical}, a spectral method was proposed for tensor principal component analysis (tensor PCA).  In this problem, we are given a $p$-index tensor which is equal to Gaussian noise plus the $p$-th tensor power of some fixed vector, and the problem is to recover this vector (called the ``spike" vector).  A slightly simpler problem is to distinguish the case where the tensor is equal to this $p$-th tensor power added to noise (called the ``spiked" case) from the case where the tensor is equal to Gaussian noise with nothing added (called the ``unspiked" case).
The spectral method gives a polynomial algorithm, with the degree of the polynomial depending on the signal-to-noise ratio.
Since the spectral problem involves a high-dimensional linear operator $H$ which can be regarded as a Hamiltonian of a quantum system on some large number of particles (the number depending on the signal-to-noise ratio), a quantum algorithm is a natural possibility and \cite{hastings2020classical} also proposed such an algorithm.  This algorithm led to a \emph{quartic} speedup over classical.  This speedup comes from combining a standard quadratic speedup from amplitude amplification with an additional quadratic speedup from a chosen input state into a phase estimation algorithm.  
To understand this chosen input state, note that the goal of the spectral algorithm is to estimate the largest eigenvalue of $H$.  For a maximally mixed input state, the probability of projecting onto the corresponding eigenstate using phase estimation is inverse in the dimension of the Hilbert space.  However, the chosen input state gives us a state with some nontrivial overlap with the eigenstate with largest eigenvalue, increasing this probability.

One may compare this chosen input state to what happens in
some cases (such as in applications in quantum chemistry) where again we are interested in computing an eigenvalue of some Hamiltonian to high precision but we have an input state (such as a Hartree-Fock state) which has some large overlap with the ground state so that phase estimation succeeds with probability close to $1$.  Here, the success probability is much smaller, but still nontrivial.

While this particular tensor PCA problem may seem somewhat specialized, there are more general reasons to be interested in it.  For one thing, the quartic quantum speedups found have been generalized to other problems with planted solutions\cite{schmidhuber2025quartic1,schmidhuber2025quartic2} recently, so this may be a more general phenomenon.
Further, in these planted problems, the solution found by a quantum computer may be efficiently verified on a classical computer, so they are a natural target problem for early noisy quantum computers, and while a quadratic speedup is perhaps too small to be useful in practice, a quartic or higher speedup may be useful.
Another reason to be interested in the problem is that the spectral algorithm for tensor PCA is closely related to some algorithms for polynomial optimization on the sphere\cite{lovitz2025hierarchy,moreno2025spectral}, so improvements (quantum or classical) for these spectral algorithms may apply more broadly.

We should be cautious in that in \cite{gupta2025classical}, a quadratic speedup was found for the classical algorithm of \cite{schmidhuber2025quartic1} \emph{in the limit of a certain parameter being large}, so that (at least in that parameter range for that problem), the quantum algorithm presents only a quadratic speedup over classical.  However, in many regimes, the quartic speedup over the best-known classical algorithm survives.   Further, Ref.~\cite{schmidhuber2025quartic1} introduced the nice concept of an ``Alice theorem" to help quantify this speedup.  Namely, their idea was that, since the spectral algorithm's correctness is proven based on certain properties of the spectrum of a random linear operator drawn from a certain distribution, and so accordingly improvements in the bounds on that  spectrum could lead to improvements in the proven runtime of the classical algorithm, they wanted to show a general way to import any such improvements to the runtime of both classical and quantum algorithms to maintain the separation.  So, an Alice theorem was defined to be a certain type of spectral bound on the linear operator, such that if proven, one could obtain certain performance guarantees on both quantum and classical algorithms.  Currently, we have a certain Alice theorem, but that theorem may be improved in future.

In this paper, we introduce a different way to speedup both quantum and classical algorithms quadratically, while keeping the same quartic separation.  Our improvement does not rely on any improvement to the Alice theorem.  Rather, we use a different algorithm.  To explain this algorithm, let us first recall how the previous algorithms work.  In all the previous spectral works cited above, the correctness of the algorithm is based on proving that in the unspiked case, the linear operator had some bound on its large eigenvalue with high probability, while in the spiked case there is, with high probability, some eigenvalue greater than this bound so that by computing the largest eigenvalue of the operator one can distinguish spiked from unspiked (called ``detection") and, with a little more work determine the vector in the spiked case (called ``recovery").  To obtain this property of the spectrum, we had to adjust some parameter, which in Ref.~\cite{hastings2020classical} and this paper we call $\nbos$, to be sufficiently large.   The dimension of the Hilbert space depends exponentially in $\nbos$.
Classically, the spectral algorithm then simply used Lanczos or similar method with a time then roughly equal to the Hilbert space dimension (other terms such as the time per Lanczos step were negligible in comparison; precisely the time was bounded by the Hilbert space dimension to the power $1+o(1)$).
The quantum algorithm then used a chosen input state which had some nontrivial overlap with this large eigenvalue.  \emph{Very} roughly speaking, the chosen input state was equal to the vector with largest eigenvalue plus some noise, and the input state has some small but nontrivial signal-to-noise ratio.  The  
nontrivial signal-to-noise ratio meant that if one took this vector and applied phase estimation there was some nontrivial chance (roughly the square-root of the Hilbert space dimension, so quadratically better than random) of obtaining the desired large eigenvalue.  Using this chosen input state and amplitude amplification gave a runtime which was the fourth-root of the Hilbert space dimension, again up to a power $1+o(1)$ due to additional time to perform phase estimation with sufficient accuracy and prepare the input state (in this paper, we neglect all details of phase estimation and state preparation as they are handled by the same techniques as before).

Now we may explain how the quadratic speedup in this paper works.  We consider a parameter $\nbos$ which is only \emph{half} as large (up to multiplicative factor $1+o(1)$) as that needed to prove correctness of the previous spectral algorithm.  Then, we (approximately) project onto the eigenspace of $H$ with eigenvalue greater than some value; that value is chosen small enough that a certain ideal state (which encodes the spike) will be included in this projection but otherwise taken as large as possible.  This projection is approximate and is handled either by Lanczos or similar sparse matrix-vector methods classically or by phase estimation quantumly.  There are other eigenvalues in this eigenspace as $\nbos$ is not large enough for this to be the only such eigenvalue (in the previous spectral algorithm, there was only one eigenvalue in this eigenspace in the spiked case, and no eigenvalues in the unspiked case).  However, the dimension of this eigenspace is quite small compared to the dimension of the full Hilbert space of $H$.  So, a random vector has only a small amplitude to survive this projection.  Thus, if we take the chosen input vector, the ideal state will survive this projection while the random part will be small in comparison.

Note that this speedup works for both classical and quantum algorithms.  The classical spectral algorithm survives as one can use Lanczos to approximately project the chosen input vector into this eigenspace, with a quadratic speedup due $\nbos$ being only half as large so that the matrix-vector operations are faster.  At the same time the quantum algorithm also gets a quadratic speedup: the signal-to-noise ratio of the chosen input state is still nontrivial and so approximate projection of this chosen input state succeeds with probability inverse in the square-root of the dimension of the Hilbert space, which can then be wrapped with amplitude amplification.

In the classical algorithm, one then considers the amplitude of the input vector after the approximate projection.  In the unamplified quantum algorithm, one considers the probability that the state survives the approximate projection by phase estimation.

In this construction, we re-use a beautiful trick introduced in \cite{schmidhuber2025quartic1}.  To explain this trick, let us recall the chosen input state of \cite{hastings2020classical}.  The first guess at a chosen input state is, in a sense, ``the only thing you could write down".  Namely, we regard the input tensor as a vector on $p$ bosons and we take the symmetrized $(\nbos/p)$-th tensor power of this to define a state on $\nbos$ bosons.  However, we run into a technical difficulty in that that state is equal to some power of the spike vector plus noise and that noise is \emph{correlated} with the Hamiltonian, making it difficult to use statistical methods.  In \cite{hastings2020classical}, a complicated procedure based on anti-concentration of Gaussian polynomials using some perturbation of this input state was introduced to avoid this.  However, in \cite{schmidhuber2025quartic1}, a much simpler trick was introduced, adding more noise to the tensor used to the define the Hamiltonian and to the tensor used to define the input state, with these noises being \emph{anti-correlated} with each other, so that the input state becomes \emph{decorrelated} from the Hamiltonian (a technical trick is that only small noise is added to that to define the Hamiltonian, where the added noise would greatly hurt the properties of the algorithm, while larger noise is added to the input state where it has less effect).
In this way, we can use more straightforward statistical methods, with only one difficulty arising from the input state being given by taking a spike, adding noise, and then taking the $(\nbos/p)$-th tensor power, which is different from the simpler cases where we take the $\nbos$-th tensor power of the spike and add noise.

While this speedup retains the quartic separation between quantum and classical,
we later give a further speedup of the quantum algorithm.  This speedup is based on replacing the projector on the given eigenspace with a sequence of projectors.  In the simplest algorithm, we consider splitting the system in half again, and apply approximate projectors on each half.  Any failure on a single half means we build that state again, and only when both halves succeed do we build the full system.  We can in fact make this recursive by building each half in terms of two quarters, etc...  This is explained in more detail in \cref{multistep}.

We will give something similar to the idea of an ``Alice theorem", by giving a general statement about the spectrum of a certain random linear operator.  In particular, we consider the number of eigenvalues which are greater than or equal to some given fraction of a bound on maximum eigenvalue.  These
bounds are currently proven if we consider the number of eigenvalues which are greater than or equal to some given fraction of \emph{our current best} bound on the maximum eigenvalue.  We state it in this form so that if the bound on largest eigenvalue improves, it gives a clear statement of what bound on the spectrum is needed.

We will mostly follow the notation of \cite{hastings2020classical} and the reader should be familiar with that paper.  We use the Hamiltonian of that paper, while other papers consider slightly different, but related, Hamiltonians.  In some cases we will avoid symmetrizing the Hilbert space and when this is done we will state it explicitly but generally we will use the symmetrized Hilbert space of that paper.  In all cases we consider even $p$, and we will largely focus on the case $p=4$ for simplicity. 

All results involving asymptotic speedups and all big-O notation assumes that we the limit
of large $N$ at fixed
 $N^{-p/4}/\lambda$.  For $p=4$, this is at fixed $\lambda N$, i.e., $\lambda \propto N^{-1}$.
We will say that a quantity is $\tilde O(N^{k})$, for some $k$ which is a polynomial in $\nbos$ if the time is bounded by $N^{k}\cdot N^{o(\nbos)}$ in this limit of large $\nbos$.  In particular, any fixed polynomial in $N$ is $\tilde O(1)$.

In \cref{background}, we give the background on the problem and previous results on the spectrum.  We straightforwardly extend some previous results to consider the number of states with eigenvalue in some range.
In \cref{algorithm}, we give both classical and quantum algorithms with quadratic speedup over previous results.
In \cref{ideal}, we show that a certain ``ideal" input state, namely the $\nbos$-th tensor power of the spike vector, has most of its amplitude in the eigenspace projected onto in the given algorithm.
In \cref{decorrnoise}, we finish the analysis of the algorithm of \cref{algorithm} by considering the actual chosen input state with noise.
In \cref{multistep}, we consider a multistep algorithm with further quantum speedup.
All the results up to there are for detection rather than recovery.  In \cref{recovery}, we discuss recovery and give a strong plausibility argument that our algorithm achieves recovery also.

Our main result, proven as a result of results in \cref{ideal,decorrnoise} is that, with high probability, algorithms \cref{classquad,qquadunamp,qquadamp} correctly answer the question of detection.

\section{Background Results}
\label{background}
Let us first recall the problem.

Let $\vsig$ be a vector in ${\mathbb R}^N$, normalized by $| \vsig| =\sqrt{N}$, chosen from some probability distribution; this is the ``signal vector".  
We pick it from a Haar uniform distribution.  The reason for the normalization that  $| \vsig| =\sqrt{N}$, which is standard in the literature, is that if we choose each component of the vector independently from a Gaussian distribution with vanishing mean and unit variance, this gives, with high probability, a norm equal to $\sqrt{N} \cdot (1+o(1))$.

Let $G$ be a real tensor of order $p$ with entries chosen from a Gaussian distribution with vanishing mean and unit variance.  Precisely, we obtain this tensor by choosing an \emph{unsymmetrized} tensor with independent entries.
Since the tensor $\vsig^{\otimes p}$ is symmetric, of course it is natural to replace $\Tspike$ by its symmetrization.  Indeed, no information can be lost by this replacement since given a tensor $\Tspike$ one can symmetrize the tensor, and then add back in Gaussian noise chosen to vanish under symmetrization to obtain a tensor drawn from the same distribution as $\Tspike$ was.
That is, the cases in which $G$ is symmetrized or not can be reduced to each other.  We restrict to the symmetric case throughout.

We let $\Tspike$ be given by
$$\Tspike=\lambda \vsig^{\otimes p} + G,$$ where $\vsig^{\otimes p}$ is defined to be the tensor with entries $$(\vsig^{\otimes p})_{\mu_1,\ldots,\mu_p}=\prod_{a=1}^p (\vsig)_{\mu_a}.$$
Here we use the notation that a subscript on a vector denotes an entry of that vector; we use a similar notation for matrices later.
We take either $\lambda=\overline \lambda$ for some known $\overline \lambda$ or we take $\lambda=0$,
depending on whether we are in the spiked or unspiked case,

A generalization of this problem is the case in which $G$ is chosen to have complex entries, with each entry having
real and imaginary parts chosen from a Gaussian distribution with vanishing mean and variance $1/2$.  We refer to this as the complex ensemble, while we refer to the case where $G$ has real entries as the real ensemble; the choice of reducing the variance to $1/2$ is a convenient normalization for later.  It is clear that since $\vsig$ is real, the case of complex $G$ can be reduced to the real case (up to an overall rescaling for the different variance) simply by taking the real part of $\Tspike$, and similarly the real case can be reduced to the complex case (again up to an overall rescaling) by adding Gaussian distributed imaginary terms to the entries of $\Tspike$.  (Of course, one might all consider the case where $\vsig$ and $G$ are both complex; this changes some numerical factors but most of the results remain the same; we omit this case.)
We assume in this paper for simplicity that $G$ is drawn from the real case, but for analysis later it is convenient to consider the complex case.

Given a tensor $T$ with even $p$, define a ``Hamiltonian" $H$ by
\be
\label{firstquantized}
H(T)  =\frac{1}{2}\sum_{i_1,\ldots,i_{p/2}}  \Bigl( \sum_{\mu_1,\ldots,\mu_p} T_{\mu_1,\mu_2,\ldots,\mu_p} 
|\mu_1\rangle_{i_1}\langle \mu_{1+p/2}| \otimes 
|\mu_2\rangle_{i_2}\langle \mu_{2+p/2}| \otimes  \ldots \otimes
|\mu_{p/2}\rangle_{i_{p/2}}\langle \mu_{p}| + {\rm h.c.} \Bigr).
\ee
This Hamiltonian acts on a Hilbert space which is $({\mathbb C}^N)^{\otimes \nbos}$, for some given quantity $\nbos$.  The basis vectors of this Hilbert space are of the form $|\mu_1\rangle_1 \otimes |\mu_2\rangle_2 \otimes \ldots \otimes |\mu_{\nbos}\rangle_{\nbos}$ where each $\mu_1$ ranges from $1$ to $N$.  The notation above indicates that it acts on $p/2$ of the tensor factors in this tensor product (namely, the $i_1,\ldots,i_{p/2}$ tensor factors), while it is implicitly diagonal in the remaining tensor factors.
The sum over $i_1,\ldots,i_{p/2}$ is from $1$ to $\nbos$, and the $i_1,\ldots,i_{p/2}$ are all distinct from each other, while the sum over $\mu_1,\ldots,\mu_p$ is over $1$ to $N$.

We may refer to this as a ``first quantized Hamiltonian", in that it is written in what would be term a first quantized form in physics.

We may restrict to the symmetric subspace of this Hilbert space.  
This restriction motivates us calling $\nbos$ the ``number of bosons", as in this subspace they are what would be called bosons in physics.
We write
$D(N,\nbos)$ to indicate the dimension of this subspace.  For $N\gg \nbos$, we can approximate $D(N,\nbos) \approx N^\nbos/\nbos!$.

In some cases later we do not restrict to the symmetric subspace.  In this case we say that we work in the ``full" Hilbert space.

Within the symmetric subspace, we can write this Hamiltonian in a so-called ``second-quantized" form:
\be
\label{sqeven}
H(T)=\frac{1}{2} \Bigl(\sum_{\mu_1,\ldots,\mu_p} T_{\mu_1,\mu_2,\ldots,\mu_p} 
\Bigl(\prod_{i=1}^{p/2} a^{\dagger}_{\mu_i} \Bigr)
\Bigl(\prod_{i=p/2+1}^p  a_{\mu_i}\Bigr)+{\rm h.c.}\Bigr).
\ee
This replacement by a second-quantized Hamiltonian is simply a convenient notation.
The operators $a^\dagger_\mu,a_\mu$ are bosonic creation and annihilation operators, obeying canonical commutation relations
$[a_\mu,a^\dagger_\nu]=\delta_{\mu,\nu}$.  We restrict to the subspace with a total of $\nbos$ bosons, i.e.,
we define the number operator $n$ by
\be
n\equiv \sum_\mu a^\dagger_\mu a_\mu,
\ee
and restrict to $n=\nbos.$

Now let us recall the spectral algorithm of \cite{hastings2020classical} for this problem (closely related to the other spectral algorithms for this problem).
This algorithm takes a tensor $\Tspike$ as input and also a scalar $\overline \lambda$ and an integer $\nbos$.  The output is a decision about whether $\lambda=\overline \lambda$ or $\lambda=0$, and, if the algorithm reports that $\lambda=\overline \lambda$, it also returns an approximation of $\vsig$ (up to an overall sign).
This algorithm is a classical algorithm which requires computing the leading eigenvalue of the Hamiltonian; quantum algorithms replace this with phase estimation with a chosen input state.  We consider this Hamiltonian in the symmetric subspace (though for quantum algorithms it may be slightly simpler to run in the full Hilbert space without symmetrization).

The quantity $\nbos$ is chosen depending upon the value of $\overline \lambda$; smaller values of $\lambda$ require larger values of $\nbos$
 so that the quantity $\emax$ will be smaller than $E_0$.
Indeed, for $E_0\geq (1+c) \emax$ for any $c>0$, the algorithm achieves recovery.
This quantity $E_0$ is a lower bound on the largest eigenvalue in the case $\lambda=\overline \lambda$, obtained by a straightforward variational principle.
The quantity $\emax$ is an upper bound that holds with high probability on the largest eigenvalue of $H$ in the case $\lambda=0$, up to subleading corrections, see \cref{eigbthm}.

\begin{algorithm}
\caption{Spectral algorithm.}  \begin{itemize}
\item[1.] Compute the eigenvector of $H(\Tspike)$ and the leading eigenvalue, denoted $\lambda_1$.

\item[2.] (Detection) If $$\lambda_1>\ecut\equiv (E_0+\emax)/2,$$ 
where $E_0=\overline \lambda (p/2)! {\nbos \choose p/2} N^{p/2}$ for even $p$, and $E_0=\overline \lambda^2 (p-1)! {\nbos \choose p-1} N^p$ for odd $p$, and
where $\emax$ is defined in theorem \ref{eigbthm}, then report that $\lambda=\overline \lambda$.  Otherwise report $\lambda=0$.

\item[3.] (Recovery) Compute the single particle density matrix (defined in Ref.~\cite{hastings2020classical}) of the leading eigenvector or of any vector in the eigenspace of eigenvalue $\geq \ecut$.  Apply the randomized recovery algorithm, given as Algorithm 1 in \cite{hastings2020classical}, to recover an approximation to $\vsig$.
\end{itemize}
\label{specalg}
\end{algorithm}

We have the following theorem from \cite{hastings2020classical}:
\begin{theorem}
\label{eigbthm}
Let $\lambda_1$ be the largest eigenvalue of $G$.
Let $$\emax=\sqrt{2J \log(N)} \nbos^{p/4+1/2} N^{p/4},$$
where 
$J$ is a scalar depends that implicitly on $p,\nbos,N$ and
tends to some function depending only on $p$ for large $\nbos,N$.  More precisely, for even $p$, $J$ 
 is equal to $(p/2)! {\nbos \choose p/2!}/\nbos^{p/2}+o(1)$ for the real ensemble and is twice that for the complex ensemble, and for odd $p$, $J$ is equal to that for the even case for $2(p-1)$.

Then,
for any $x$, assuming Assumption 1 of Ref.~\cite{hastings2020classical} (this assumption is implicit in our scaling limit),
\be
\prob[\lambda_1\geq x] \leq \exp\Bigl(-\frac{x-\emax}{\xi}\Bigr),
\ee
with
\be
\xi=\frac{\sqrt{J} \nbos^{p/4-1/2} N^{p/4}}{\sqrt{2\log(N)}}.
\ee

So, for any $\dem$ which is $\omega(\xi)$, with high probability $\lambda_1\leq \emax+\dem$.
\end{theorem}

Now consider the eigenvectors and eigenvalues of $H(\Tspike)$.  
For any $q$, for any symmetric tensor $T$ of order $q$, let $|T\rangle$ be
the vector on $q$ qudits (each of dimension $N$) with amplitudes given by 
the entries of the tensor in the obvious way:
$$|T\rangle=\sum_{\mu_1,\ldots,\mu_{q}} T_{\mu_1,\ldots,\mu_{q}}
|\mu_1\rangle \otimes \ldots \otimes |\mu_{q}\rangle.$$
This vector is only normalized if $|T|=1$.
So, $\Psisig\equiv N^{-\nbos/2} |\vsig^{\otimes \nbos}\rangle$ is a normalized vector.

We have the following simple property:
\begin{lemma}
\label{collect}
Let $\lambda_1$ be the largest eigenvalue of $H(\Tspike)$.  
Then, $\lambda_1 \geq E \equiv \langle \Psisig | H(T) |\Psisig \rangle$.
\begin{proof}
Immediate from the variational principle.
\end{proof}

\end{lemma}

Thus, we have two different scaling results.  We have, from \cref{eigbthm}, that if $\lambda=0$, then with high probability the largest eigenvalue is bounded by $\nbos^{3/2 + o(1)}$.  On the other hand, for nonzero $\lambda$, the largest eigenvalue is proportional to $\lambda \nbos^2$.  Considering the point at which these two values cross gives the scaling of the needed $\nbos$.  
Let $\nboseq$ be the point at which both values cross.  Indeed, it suffices to chose $\nbos=\nboseq \cdot (1+c)$ for any $c>0$ (indeed, we may take $c=o(1)$) for the algorithm to work.

Now, it is certainly possible that the bounds in \cref{eigbthm} might be improved later, both in terms of prefactor and possibly even in terms of scaling with $\nbos$.  In order to state our results on the new algorithm later in a way that will survive changes in these bounds, let us introduce a \emph{scaling exponent}, denoted $\gamma$.  We will assume that the largest eigenvalue is, with high probability bounded by $\emaxfuture\cdot (1+o(1)$, where we define
\be
\emaxfuture=f(N) \nbos^{\gamma},
\ee
for some $f(N)$.
Call this the scaling assumption on the largest eigenvalue.
Here we right $\emaxfuture$ rather than $\emax$ to emphasize that this is some possible bound proven in the future, improving on existing bounds.
The bounds proven in \cref{eigbthm}, hold for $f(N)=\sqrt{2J \log(N)} N$ in the case $p=4$ and
\be
\gamma=3/2.
\ee
We also need later an assumption on the probability that a random eigenvalue is some constant fraction larger than $\emaxfuture$.  That is, let $\Pi^>(\nbos,x)$ project onto the eigenspace with eigenvalue $\geq x \emaxfuture$ for scalar $x$.  Then, define $P^>(\nbos,x)=\expec[{\rm tr}(\Pi^>)/{\rm tr}(I)]$, where $I$ is the identity matrix so ${\rm tr}(I)$ is the total dimension of the Hilbert space.
We will assume that this obeys a \emph{scaling function}
\be
P^>(\nbos,x) \leq N^{-g(x) \nbos}.
\ee
(A stronger assumption is that this equation holds with equality up to subleading corrections; we conjecture that this is true for the function $g(x)$ in \cref{fcbound} below if $\emaxfuture$ is a tight bound on the largest eigenvalue up to subleading corrections.)
For example, at $x=1$, this probability is $1/N^{\nbos+o(1)}$.  Indeed, ${\rm tr}(I)$ is equal to $N^{\nbos+o(1)}$, where one should
note that a term such as $\nbos^\nbos$ is, by the assumption that $N$ is taken large first, equal to $N^{o(1)}$ and expect ${\rm tr}(\Pi^+)$ is small with only $N^{\o(1)}$ eigenvalues that large.
We claim that the bounds proven in lemma 4, equation (26)\footnote{This lemma is essentially a reproof of the matrix Chernoff bound in this particular case.} of Ref.~\cite{hastings2020classical} lead to
\be
\label{fcbound}
g(x)=x^2+o(1).
\ee
We call this bound a \emph{bound on the density of states.}
Indeed, the way that these bounds were used in lemma 4 of Ref.~\cite{hastings2020classical} is as follows.  These are bounds on the expectation value of a quantity $Z(\tau)=\expec[{\rm tr}(\exp[\tau H])]$ for some scalar $\tau$, for $\lambda=0$; the quantity $Z$ depends also on $N,\nbos,p$ but we only make the $\tau$ dependence explicit here.  We have the trivial lower bound that
${\rm tr}(\exp[\tau H])\geq \exp(x \emax \tau) {\rm tr}(\Pi^+)$ for any $H$, and so
$Z(\tau)\geq \exp(x \emax \tau) P^>(\nbos,x) {\rm tr}(I)$, so 
$$P^>(\nbos,x) \leq Z(\tau) \exp(-x \emax \tau)/{\rm tr}(I).$$
In \cite{hastings2020classical}, we considered $x=1$, and then a choice of $\tau$ was optimized over to minimize the right-hand side for given $\emax$; in this way, we were able to prove that for the given $\emax$ the right-hand side was less than $1$, and in this way we were able to prove (after one more steps) that the largest eigenvalue was, with high probability, not much larger than $\emax$.
So, one may do a similar thing here, taking $x<1$, and seeing how the right-hand side behaves, giving the claimed result \cref{fcbound}.

\section{The Algorithm}
\label{algorithm}
In this section, we give the quantum and classical algorithms which realize a quadratic speedup over previous results.

First, we recall the trick of \cite{schmidhuber2025quartic1}.
Given the tensor $\Tspike$ in the problem, define two new tensors $\Tplus,\Tminus$ by adding more noise to the problem, and then applying a rescaling.
Let $G'$ be some other Gaussian noise chosen from the same distribution as $G$.
Let $\Tplus=\frac{T+\zeta H'}{\sqrt{1+\zeta^2}}$ and let $\Tminus=\frac{T-\zeta^{-1} H'}{\sqrt{1+\zeta^2}}$, for some scalar $\zeta$.
By adding this extra noise, we find that now the Gaussian noise in $\Tplus$ is \emph{uncorrelated} from that in $\Tminus$.
That is, the distribution of $\Tplus,\Tminus$ is the same as if we had chosen
\be
\Tplus=\lambda^+ \lambda \vsig^{\otimes p} + G^+,
\ee
and
\be
\label{Tmdef}
\Tminus=\lambda^- \vsig^{\otimes p} +G^-,
\ee
where $G^+,G-$ are independent of each other and both chosen from the same distribution as $G$.
Here
\be
\lambda^{\pm}=\lambda (1+\zeta^{\pm 2})^{-1/2}.
\ee
Similarly, let $\overline \lambda^{\pm}=\overline \lambda (1+\zeta^{\pm 2})^{-1/2}$.

We will use choose $\zeta=1/\log(N)$.  
Indeed, any choice of $\zeta$ would work so long as it is $o(1)$ but decays subpolynomially in $N$.
Then, $\lambda^{+}$ is only slightly smaller than $\lambda$, while $\lambda^-$ is smaller than $\lambda$ by a polylogarithmic amount.
This is crucial: reducing $\lambda^+$ can force us to increase $\nbos$ and the algorithm runtime is exponential in $\nbos$, so we minimize the reduction in $\lambda^+$, while reducing $\lambda^-$ has only a small effect.

For  use later, it is convenient to add \emph{more} noise to $\Tminus$.  This noise is pure imaginary.  We do this so that $G^-$ is chosen from a complex distribution, while $G^+$ is chosen  from the real distribution.

Remark: in practice, likely one would want to not add this pure imaginary noise to $\Tminus$, and indeed likely one would want to simply take $\Tminus=\Tplus=\Tspike$ so that no additional noise is added.  However, adding this noise simplifies the analysis.

We now give the classical algorithm for detection.
The algorithms rely on a chosen input state
\be
\psiinput=\frac{1}{|\Tminus|^{\nbos/4}} |\Tminus\rangle^{\otimes \nbos/4}.
\ee
Assuming $\gamma=3/2$ and \cref{fcbound} on the density of states holds,
the quantity $\nbos$ needed will be only half as large (up to multiplicative $o(1)$ corrections) as the $\nboseq$ needed in \cref{specalg}.

The algorithms all also rely on certain approximate projectors.  In an idealized world, this projector would simply be a projector onto an eigenspace with eigenvalue greater than or equal to some given value.  However, in reality, we need to replace this with an approximate projector.  For the classical algorithm, this is because we would use some form of Lanczos to compute the approximate projection of some state onto some eigenspace and we have a finite number of Lanczos steps.  For the quantum algorithm, this is because we need to do some form of phase estimation to implement the projector, phase estimating the eigenvalue, determining whether it exceeds some cutoff, and then uncomputing the eigenvalue.  For our purposes, we need a few properties of the projector.  The projector is diagonal in the eigenspace of $H$.  There is some ``lower cutoff" $E_{\rm lower}$ and ``upper cutoff" $E_{\rm upper}$ with $E_{\rm lower}<E_{\rm upper}$.  In the eigenspace of eigenvalues greater than the upper cutoff, the projector is close to $1$; indeed, we require that its difference from $1$ be much smaller than $N^{-\nbos}$ so any error there is negligible.  In the eigenspace of eigenvalues smaller than the lower cutoff, the projector is similarly close to $0$; its difference from $0$ is much smaller than $N^{-\nbos}$.  Finally, we want the difference between upper and lower cutoff to be polynomially small in $N$, for example $1/N$ would suffice\footnote{We claim that these approximate projectors can be implemented classically in a time bounded by the Hilbert space dimension times a polynomial in $N,\nbos$.  Indeed, use sparse matrix-vector methods to implement the approximate projector.  A polynomial of degree at most polynomial in $N,\nbos$ can implement such an approximate projector.  Similarly, the time to implement the projector quantumly is polynomial in $N,\nbos$ using phase estimation.}.  From now on, when we say that we have an approximate projector onto a certain eigenspace with eigenvalue greater than something, this eigenvalue can be taken to be the upper cutoff.

Finally, later we will prove results either lower bounding or upper bounding the projection of some state onto some eigenspace; if we prove some lower bound, then this is lower bounding the projection onto the eigenspace with eigenvalue greater than the upper cutoff, while if we prove an upper bound, then this is upper bounding the projection onto the eigenspace with eigenvalue greater than the lower cutoff.  That way, in both cases it implies some lower or upper bound on the norm of the state after applying the approximate projector.  We will not state this explicitly later.  The difference between lower and upper cutoff is small enough that it is negligible when applying these results.

For use later, we denote this approximate projector $\tilde \Pi$.  If needed later, we will further write this projector as $\tilde \Pi(\nbos)$ to indicate the particular $\nbos$.

We also define a projector $\Pis$ onto the symmetric subspace of the full Hilbert space.  In practice, symmetrization may not be necessary but it is convenient to consider.

The algorithm depends on a quantity $\normmin$, from \cref{normminlemma}.  This is an estimate of the norm-squared of the projection onto the given eigenspace in the spiked case.  It is $N^{-\nbos/2}/\tilde O(1)$.

\begin{algorithm}
\caption{Classical algorithm with quadratic speedup}  \begin{itemize}
\item[1.] Take the chosen input state, project it onto the symmetric subspace, and compute its approximate projection onto the eigenspace of $H(T^+)$ with
eigenvalue greater than or equal to 
$(1-c') \overline \lambda^+ N^2 \nbos (\nbos-1)$.  The quantity $c'$ is $o(1)$ and is the same $c'$ as in the statement of
\cref{eigenfluc}.

\item[2.] If the projection onto this eigenspace has norm-squared greater than or equal to $\normmin$ then report $\lambda=\overline \lambda$, otherwise report $\lambda=0$.

\item[3.] If reporting $\lambda=\overline \lambda$, compute the single particle density matrix of this projection and apply the same randomized recovery algorithm to approximately recover $\vsig$.
\end{itemize}
\label{classquad}
\end{algorithm}

This classical algorithm relies on approximate projection into an eigenspace.  This can be computed using Lanczos methods in time $\tilde O(N^\nbos)$, giving the quadratic speedup for $\nbos$ half as large as previously, using space $\tilde O(N^\nbos)$.

We also give the quantum algorithm.  We begin with the unamplified version, \cref{qquadunamp}.

\begin{algorithm}
\caption{Quantum algorithm, unamplified version}
\begin{itemize}
\item[1.] Take the chosen input state, project it onto the symmetric subspace, and use phase estimation to approximately project it onto the eigenspace of $H(T^+)$ with
eigenvalue greater than or equal to 
$(1-c') \overline\lambda^+ N^2 \nbos (\nbos-1)$.  The quantity $c'$ is $o(1)$ and is the same $c'$ as in the statement of
\cref{eigenfluc}.  Repeat this step until it succeeds, or until it is repeated $c''/\normmin$ times, for some $c''$ which is $\omega(1)$ but $\tilde O(1)$.

\item[2.] If the projection onto this eigenspace succeeds, report 
$\lambda=\overline \lambda$, otherwise report $\lambda=0$.

\item[3.] If the projection does succeed, compute the single particle density matrix of this projection and apply the same randomized recovery algorithm to approximately recover $\vsig$.
\end{itemize}
\label{qquadunamp}
\end{algorithm}
Note that if the projection onto the eigenspace has norm squared at least $\normmin$, then it succeeds with probability $1-o(1)$.

The quantum algorithm uses phase estimation to compute the approximate projection.  The techniques for doing this to sufficient accuracy are the same as in \cite{hastings2020classical}.
It is essential that the error be small compared to $N^{-\nbos}$.
The runtime of the algorithm is then equal to the number of trials, multiplied by the time for each trial.
Assuming $\gamma=3/2$ and \cref{fcbound} holds, we will see that this time is $\tilde O(N^{\nbos/2})$ so it is quadratically faster than \cref{classquad}.
We can then wrap the algorithm in amplitude amplification, to obtain \cref{qquadamp}.

\begin{algorithm}
\caption{Quantum algorithm, amplified version}
\begin{itemize}
\item[1.] Apply amplitude amplification to step {\bf (1.)} of \cref{qquadamp} so that if the projection succeeds with probability $\normmin$, one obtains, with probability $1-o(1)$, a vector with large norm on the given eigenspace.

\item[2.] If amplitude amplification succeeds, report 
$\lambda=\overline \lambda$, otherwise report $\lambda=0$.

\item[3.] If amplitude amplification does succeed, compute the single particle density matrix of the resulting state and apply the same randomized recovery algorithm to approximately recover $\vsig$.
\end{itemize}
\label{qquadamp}
\end{algorithm}

To understand this algorithm, note that, in the spiked case, the quantity $\overline\lambda^+ N^2 \nbos (\nbos-1)$ is the expectation value of $H(T^+)$ for on the state
\be
\psiideal=\frac{1}{N^{\nbos/2}}|\vsig\rangle^{\otimes \nbos'},
\ee
up to small fluctuations due to the Gaussian noise.  In \cref{ideal}, we consider this in more detail, considering the fluctuations in energy for this ``ideal" input state, but the result is to show that such an ideal input state will have a large probability to be in the given eigenspace; see \cref{eigenfluc}.
In \cref{decorrnoise}, we show that, with high probability, the input state $\psiinput$ has probability at least $\normmin$ to be in the given eigenspace.  Heuristically, this is because $\psiinput$ is the ideal input state, which is in the given eigenspace, plus some noise, and we can bound the amount of noise.
So, this will show that in the spiked case, the algorithm correctly reports $\lambda=\overline \lambda$ with high probability, regardless of the value of $\nbos$.

Now consider the unspiked case.  We first need the following \cref{uniform} showing that the average of the input state over noise is maximally mixed after projecting to the symmetric subspace.
Recall that $\Pis$ is the projector onto the symmetric subspace so it is the identity matrix on this subspace on $\nbos$ bosons.
\begin{lemma}
\label{uniform}
Consider
$\expec[\Pis (|G\rangle\langle G|)^{\otimes \nbos/p}\Pis],$
where we consider the expectation value over $G$  drawn from the symmetric, complex distribution.  This expectation value equals
$(\nbos/p)! \Pis.$ 

Similarly, consider the expectation value
$\expec[\frac{1}{|G|^{2\nbos/p}} \Pis (|G\rangle\langle G|)^{\otimes \nbos/p}\Pis].$  This is the expectation value of a \emph{normalized} state.  This
equals $(1+o(1)) \rhosymmmm,$ where $\rhosymmmm$ is the maximally mixed state on the symmetric subspace.
\begin{proof}
The quantity$|G\rangle\langle G|)^{\otimes \nbos/p}$ is a polynomial of degree $\nbos/p$ in $G$ and of degree $\nbos/p$ in the complex conjugate $\overline G$.  We may compute the expectation value using Wick's theorem as a total of $(\nbos/p)!$ terms.  Each such term is given by applying some permutation $\pi$ to the outgoing indices of the state $(I_{p,{\rm symm}})^{\otimes \nbos/p}$, where $I_{p,{\rm symm}}$ is the identity matrix in the symmetric subspace on $p$ bosons.  (Outgoing indices refers to those indices on the ket of this density matrix.)  
All terms gives the same contribution to 
$\Pis (|G\rangle\langle G|)^{\otimes \nbos/p}\Pis$, since $\Pis$ is given by averaging over permutations of indices.  So,
$\expec[\Pis (|G\rangle\langle G|)^{\otimes \nbos/p}\Pis]=(\nbos/p)!
\Pis (I_{p,{\rm symm}})^{\otimes \nbos/p} \Pis =(\nbos/p)! \Pis I_p^{\otimes \nbos/p} \Pi_s $, where $I_p$ is the identity matrix in the full space on $p$ bosons, as the non-symmetric components of $I_p$ are annihilated by $\Pis$.  However, $I_p^{\otimes \nbos/p}$ is the identity matrix on the full subspace so then
$\Pis (I_{p,{\rm symm}})^{\otimes \nbos/p} \Pis = \Pis.$

The result for the normalized follows since the norm concentrates about its mean so the fluctuations in norm may be ignored.  Further, since the norm of the state is uncorrelated from the normalized state, the average of the normalized state is still proportional to maximally mixed state.
\end{proof}
\end{lemma}

Now we may show
\begin{lemma}
Suppose the scaling assumption  on the largest eigenvalue holds with $\gamma=3/2$ and suppose \cref{fcbound} holds.
Then
for $\nbos=(1/2) \nboseq \cdot (1+c''')$, for some $c'''=o(1)$, in the unspiked case the algorithm correctly reports that it is unspiked with high probability.
\begin{proof}
From the scaling assumption on the largest eigenvalue, at $\nbos=(1/2) \nboseq$, we are projecting onto an eigenspace with eigenvalue at least $(1/\sqrt{2})\emax$ up to subleading corrections.  From \cref{fcbound}, the dimension of this eigenspace is equal to the dimension of the symmetric space times $\tilde O(1/N^{-(1/2) \nbos})$.
From \cref{uniform}, the expectation value over noise in $G^-$ of the quantum expectation value of the approximate projector is
$O(1/N^{-(1/2) \nbos \cdot (1-o(1))})$.
Here there are two expectation values: we take the expectation value over noise in $G^-$ and we take the expectation value of the approximate projector in the given input state.
So, with high probability, the norm squared of the projection of the input state onto the given eigenspace is
$O(1/N^{-(1/2) \nbos \cdot (1-o(1))})$.  This is for $\nbos=(1/2) \nboseq$ and in this case we have the projection equal to $\normmin$ up to subleading corrections.
So, for $\nbos=(1/2) \nboseq \cdot (1+c''')$, we may choose $c'''=o(1)$ so that the projection is smaller than $\normmin$ with high probability.
\end{proof}
\end{lemma}

Let us emphasize the use of our scaling assumptions here.  They play the same role as the ``Alice theorem" of \cite{schmidhuber2025quartic1}.
That is, our tightest proven bound on the largest eigenvalue (used to prove the correctness of the algorithms of \cite{hastings2020classical}, and similar bounds used to prove the correctness of the algorithms of \cite{wein2019kikuchi}) obeys the scaling assumption with $\gamma=3/2$, and the proof of these bounds can be used to prove \cref{fcbound} for the given $\emaxfuture=\emax$.  Thus, relative to our best existing bounds for the speed of the previous spectral algorithms, these assumptions do hold, and so indeed we may claim (once some other results in this paper are proven) a quadratic speedup for the performance of both classical and quantum algorithms relative to the best previously proven performance.  The reason for expressing the bounds in this form is that if someone does somehow tighten the bounds on the spectrum, then we will still have a quadratic speedup so long as the scaling assumption holds with $\gamma=3/2$ and so long as \cref{fcbound} for these improved bounds on the largest eigenvalue.

\section{Success Probability for Ideal Input State}
\label{ideal}
Now we consider the state
 $\psiideal=\frac{1}{N^{\nbos/2}}|\vsig\rangle^{\otimes \nbos}$, and address the question: if we measure the energy of this state using $H(\Tplus)$,
 what is the probability distribution of outcomes?  In physics terms, this would be called the spectral weight.

Let $P_{n,k}=n!/(n-k)!$, denoting the number of ways to choose $k$ items out of $n$ items, where the order in which the $k$ items are chosen matters.
We have $P_{n,2}=n(n-1)$.

Then, one may verify that
\be
\langle \psiideal | H| \psiideal \rangle=P_{\nbos',2} (\vsig)_\mu (\vsig)_\nu (\vsig)_\sigma (\Tplus)_{\mu\nu\rho\sigma}.
\ee

For simplicity of writing the results, let us rotate the basis for the $N$-component vectors so that $(\vsig)_0=\sqrt{N}$ and all other components of $\vsig$ vanish.  (Of course, one does not know this basis when running the algorithm, but since the Gaussian distribution  is rotation invariant, this does not change the probability distribution of the quantities we calculate here.)
In this case, we have
\be
\label{hexpec}
\langle \psiideal | H| \psiideal \rangle=P_{\nbos',2} (\Tplus)_{0000}.
\ee

Also, with the same rotation
\begin{eqnarray}
\label{hsqexpec}
\langle \psiideal | H^2| \psiideal \rangle &=& P_{\nbos',4} (\Tplus)_{0000}^2 \\ \nonumber
&&+ 4 P_{\nbos',3} (\Tplus)_{000\alpha} (\Tplus)_{000\alpha} 
\\ \nonumber &&+ 2 P_{\nbos',2} (\Tplus)_{00\alpha\beta}(\Tplus)_{00\alpha\beta} .
\end{eqnarray}
This result can be understood most simply in terms of the first-quantized form of the Hamiltonian.
The right-hand side in the first line above is where the indices $i_1,i_2$
 of the first $H$ (reading from left to right in $H^2$) are all distinct from those in the second $H$; we may say they ``act on different bosons".
In the second line, one index in the first $H$ agrees with an index in the second $H$.
In the third line, both indices in the first $H$ agree with an index in the second $H$, with the first of these lines having those indices in different blocks of four bosons, while the second of these lines has them in the same block.

So, it is immediate that
$\expec[\langle \psiideal | H| \psiideal \rangle]=\lambda N^2 P_{\nbos',2}$, where we take the expectation over choices of $G^+$, since the expectation of $G^+_{0000}$ vanishes.
Similarly, 
\begin{eqnarray}\expec[\langle \psiideal | H| \psiideal \rangle^2]&=&\Bigl(\lambda N^2 P_{\nbos',2}\Bigr)^2+(1+\zeta^2)
\\ \nonumber
&=& \expec[\langle \psiideal | H| \psiideal \rangle]^2+(1+\zeta^2) .
\end{eqnarray}
The second term is asymptotically negligible compared to the first, so:
\begin{lemma}
In the limit of large $N$ at fixed
 $\lambda N$, with $\zeta=O(1)$, 
 with high probability the quantity
$\langle \psiideal | H| \psiideal \rangle$ is within a factor $1+o(1)$ of its mean, i.e.,it is in the interval $[(1-o(1)) \lambda N^2 P_{\nbos',2},(1+o(1)) \lambda P_{\nbos',2}]$.
\end{lemma}

Now consider the expectation value of $H^2$.
\begin{eqnarray}
\expec[\langle \psiideal | H^2| \psiideal \rangle] &=& \lambda^2 N^4 P_{\nbos,2}^2+
 P_{\nbos',4} O(1)
  +P_{\nbos',3} O(N)
 +P_{\nbos',2} O(N^2) \\ \nonumber
 &=& \expec[\langle \psiideal | H| \psiideal \rangle]^2+
 P_{\nbos',4} O(1)
  +P_{\nbos',3} O(N)
 +P_{\nbos',2} O(N^2).
\end{eqnarray}

So, 
\be \expec[\langle \psiideal | H^2| \psiideal \rangle] -\expec[\langle \psiideal | H | \psiideal \rangle^2] =
P_{\nbos',4} O(1)
  +P_{\nbos',3} O(N)
 +P_{\nbos',2} O(N^2).
 \ee
Thus, in the limit of large $N$ at fixed $\lambda N$, the first and second terms on the right-hand side of this equation are negligible compared to
$\expec[\langle \psiideal | H| \psiideal \rangle]^2$
while the last term has the same scaling in $N$.  However, the last term is negligible compared to $\expec[\langle \psiideal | H| \psiideal \rangle]^2$ in the limit of large $\nbos'$.
Hence,
\begin{lemma}
\label{eigenfluc}
Take the limit of large $N$ at fixed
 $\lambda N$, with $\zeta=O(1)$, and then the limit of large $\nbos'$.  Then, with high probability, if $\psiideal$ is projected into an eigenbasis of $H$, the eigenvalue is within $1+o(1)$ of the mean, and hence the projection of $\psiideal$ onto the
 eigenspace with eigenvalue
$\geq 1-c' \lambda N^2 P_{\nbos',2}$, for some $c'=o(1)$, has norm $1-o(1)$.
 (Remark: there are of course two sources of randomness here, one in the choice of random Gaussian defining $H$ and the other in the outcome of the measurement.)
\end{lemma}
Here we have just used the second moment of $H$ to bound fluctuations in $H$; we could use higher moments to get tighter tail bounds if desired.

\section{Success Probability for Chosen State}
\label{decorrnoise}
Recall that we have chosen $\Tminus$ to be from the complex distribution even though $\Tplus$ is from the real distribution.  This is to simplify calculations in this section.  It is likely not needed.  At worst, it gives up a factor in the success probability which is a constant factor raised to the power $\nbos$, which is negligible compared to polynomials in $N^{\nbos}$.

Our  approach is analogous to the Gaussian version of theorem 38 of  \cite{schmidhuber2025quartic1}, with the Gaussian version being used later in that paper.  Indeed, we could perhaps just directly import their result, 

Consider
\be
\psiinput=\frac{1}{|\Tminus|^{\nbos/4}} |\Tminus\rangle^{\otimes \nbos/4}.
\ee

Let $\psiinputun$ denote the unnormalized input state
\be
\psiinputun=\|\Tminus\rangle^{\otimes \nbos/4}.
\ee
The norm
$|\Tminus|$ concentrates about its mean in the limit of large $N$.  This mean is 
$(\lambda^2 N^4 + s^+ N^4)^{1/2}$.
 So $\psiinput \approx (\lambda^2 N^4 + s^+ N^4)^{-\nbos/8} \psiinputun.$
 
For any $v$ we have
\be
\label{expecnaive}
\expec[ |\langle v | \psiinputun \rangle|^2] \geq |\langle v | \vsig^{\otimes \nbos} \rangle|.
\ee
To see this, note that the left-hand side is a polynomial in the random Gaussian $G^-$, with all coefficients in the polynomial positive.  So, in expectation it must be greater than the zeroth order term, which is given on the right-hand side.
This holds in particular for $v$ being the (approximate) projection of $\psiideal$ onto the eigenspace with eigenvalue greater than or equal to the upper cutoff.

So,
\be
\label{expecnaive2}
\expec[ \langle v | \psiinputun \rangle|^2] \geq (\lambda N^2)^{\nbos/4} +{\rm negligible \, error},
\ee
where the negligible error is due to small projection of the ideal input state onto the eigenspace below this cutoff and due to small error in the approximate projector.
\cref{expecnaive}, combined the concentration of $|\Tminus|$ about its mean, gives a lower bound on the average of probability of success when measuring
$\tilde \Pi$ on the input state.  
We will pick $\normmin$ to be slightly smaller this value, namely equal tot this value times $1/\tilde O(1)$.

However, we want to ensure that for a random choice of $G^-$, the success probability is, with high probability, not much smaller than its mean.  That is, we want to bound fluctuations.

Here, a simple second moment method works.
So, we will bound
$$\expec[ |\langle v | \psiinputun \rangle|^4]=\expec[ \langle v |\psiinputun\rangle \langle v |\psiinputun\rangle \langle \psiinputun| v \rangle  \langle \psiinputun| v \rangle .
$$ 
We have explicitly written the right-hand side out as a product of four factors, the first two being a polynomial in $G^-$ and the second two in its complex conjugate $\overline G^-$.  We compute the expectation value using Wick's theorem, contracting factors of $G^-$ with factors of $\overline G^-$.

Suppose we consider some number, $k_{13}$ of factors of $G^-$ in the first one with the same number of factors of $\overline G^-$ in the third one, and similarly contract $k_{24}$ factors between the second and third, while contracting no factors between the first and fourth or second and third.
There are ${\nbos/4 \choose k_{13}}^2 k_{13}! {\nbos/4 \choose k_{24}}^2 k_{24}!$ possible such contractions.
To bound each such contraction, consider first  the case that all such contractions involve the ``same" boson.   Here, we are working in the full Hilbert space and the ``same" boson means contracting, e.g., a factor of $G^-$ in the $i$-th tensor factor with the factor of $\overline G^-$ also in the $i$-th tensor factor.
Then the contribution to the expectation value of the product of the first and third factor is equal to
the expectation value of the projector $|v\rangle\langle v|$ on an input state which is maximally mixed on $k_{13}$ sets of four bosons (times an overall normalization of $N^4$ on each set) and a product state of $\vsig^{\otimes 4}$ on the remaining $\nbos/4-k_{13}$ sets of four bosons bosons, and similarly for the product of the second and fourth factor.  
We may write this as
more generally as ${\rm tr}(|w\rangle\langle w|\rho)$ where $w=v \otimes v$ and $\rho = \rho_{13} \otimes \rho_{24}$ with $\rho_{13},\rho_{24}$ being those two given input states.
Note that $|w\rangle\langle w|$ is a pure state, so its trace with $\rho$ is bounded by the operator norm of $\rho$ which is
$|\vsig|^{\nbos-4k_{13}} |\vsig|^{\nbos-4k_{24}}$.
For large $N$, if $k_{13}+k_{24}>0$ this is asymptotically negligible compared to the term with $k_{13}=k_{24}=0$, and even the sum over all such terms with $k_{13}+k_{24}>0$ is asymptotically negligible as the number of terms is only exponential in $\nbos$ and independent of $N$.
Now consider the more general case where some contractions involve ``different" bosons, as well as the case in which there are some number of factors $k_{14},k_{23}$ contracted between first and fourth or second and third term.
Any such term can be written as ${\rm tr}(|w\rangle\langle w| {\rm SWAP} \rho)$.  This again is bounded by the operator norm of $\rho$ and so again is asymptotically negligible for $k_{13}+k_{24}>0$.

 So,
 \begin{lemma}
 \label{normminlemma}
 With high probability, for
 $$\normmin=\frac{1}{\tilde O(1)}    N^{-\nbos/2}/,$$
 in the spiked case, the projection by $\tilde \Pi$ succeeds with probability at least $\normmin$.
\end{lemma}

\section{Multi-Step Algorithm}
\label{multistep}
We now define the multistep quantum algorithm which gives a sixth-power speedup over classical.

Consider the system on $(1/2) \cdot (1+o(1)) \nboseq$ bosons used in the quantum algorithm of \cref{qquadamp}.  Let $\nbos(0)$ denote this number of bosons.  We subdivide this system into two halves.  We then subdivide each of those haves in two halves again, for some number, $k$ of steps so that it is divided into $2^k$ subsystems each of $1/2^k$ the size.  Let us write $\nbos(j)$ to denote the number of bosons after $j$ steps of subdivision.
The number $k$ is chosen small enough so that $\nbos/2^k$ is still much larger than $1$ so that our asymptotic analysis for large boson number applies.  As $k$ increases, the speedup asymptotically approaches the sixth-power.

If $\nbos/4$ is not a multiple of $2^k$, we replace this subdivision with an approximate subdivision in an obvious way: each time one subdivides a system in two halves, do the subdivision so that the two halves both have $0 \mod 4$ bosons and so that they either are the same size or differ by four bosons.  
For notational simplicity, we omit this detail in giving the system sizes below.

Now let us recursively define a \emph{preparation} algorithm for a system with some number $\nbos(j)$ bosons.  We will call this algorithm $A_j$, where $A_j$ prepares the system on $\nbos/2^j$ bosons.
The algorithm $A(k)$ will simply take the input state $\frac{1}{|\Tminus|^{\nbos(k)/4}} |\Tminus\rangle^{\otimes \nbos(k)/4}$, i.e., it is the same as the previous $\psiinput$, except on $\nbos(k)$ bosons.

For $j<k$, let us first describe an unamplified version of the algorithm $A(j)$.  This prepares takes the tensor product of the states on the two subsystems at level $j+1$ (each such subsystem prepared with $A(j+1)$), symmetrizes the state, and then acts with $\tilde \Pi(\nbos(j))$, repeating these steps until it succeeds.
The time to run this algorithm $A(j)$ is then equal to the time to run $A(j+1)$ times the inverse probability that the projection succeeds, up to subleading corrections.  (Of course, we must run $A(j+1)$ \emph{twice}, to prepare both subsystems, but this is simply a constant factor of $2$.) 
We can then apply amplitude amplification to this algorithm so that the time is 
 to run this algorithm $A(j)$ is then equal to the time to run $A(j+1)$ times the square-root of the inverse probability that the projection succeeds, up to subleading multiplicative corrections.   Call this probability that the projection succeeds $P(j)$.
Thus, the time to run $A(0)$ is equal to $\prod_{j=0}^{k-1} P(j)^{-1/2}$, up to subleading multiplicative corrections.

Let $Q(j)$ be the probability that the projection $\tilde \Pi(j)$ would succeed on $\nbos(j)$ bosons on the state 
 $\frac{1}{|\Tminus|^{\nbos(j)/4}} |\Tminus\rangle^{\otimes \nbos(j)/4}$ in the \emph{unspiked case}; this is in contrast to $P(j)$ which is a success probability on a state which has been prepared by the previous approximate projections.  This probability $Q(j)$ is $\leq N^{-(1-o(1))\cdot \nbos(j)/2^{j+1}}$ assuming the scaling assumption on the density of states; again we emphasize that relative to our best bound on the largest eigenvalue of the spiked Hamiltonian, we have this bound on the density of states and so the scaling assumption is needed only if the bound on largest eigenvalue is improved.
This estimate of the probability is based solely on the number of states with the given eigenvalue, i.e., the dimension of the Hilbert space approximately projected onto by $\tilde \Pi(j)$.  

Now consider instead the expectation of some operator
such as
$\ldots \Bigl(\tilde \Pi(j+1) \otimes \tilde \Pi(j+1) \Bigr) \tilde \Pi(j) \Bigl(\tilde \Pi(j+1) \otimes \tilde \Pi(j+1) \Bigr)\ldots$.
Such an operator represents the effect of previously applying other operators $\tilde \Pi(j+1),\ldots$ and so this expectation gives the probability that all
the projections succeed, i.e, it is $P(j) P(j+1)^2 P(j+2)^4 \ldots$.
Clearly, this expectation value on the maximally mixed state is bounded by $Q(j)$.  So, $P(j) P(j+1)^2 P(j+2)^4 \ldots \leq Q(j)$.

We claim that the fluctuations in $P(j)$ and $Q(j)$ are negligible, in a way a that we now quantify.  The argument we will use is not based on usual Milman or L\'evy arguments about concentration of measure on a high-dimensional sphere.  Instead we use a second moment method.
Consider some arbitrary positive semi-definite operator $O$.  In our case, it is a product of the approximate projectors $\tilde \Pi$.
Consider the expectation value of the square of the quantum expectation value of this operator in the \emph{unnormalized} state $|\Tminus\rangle^{\otimes \nbos(j)/4}$ in the full Hilbert space, i.e., $ \expec[\langle (\Tminus)^{\otimes \nbos(j)/4}| O |(\Tminus)^{\otimes \nbos(j)/4} \rangle^2]$.
  This expectation value can be computed using Wick's theorem, with $(2 \nbos(j)/4)!$ possible contractions.  One such contraction is equal ${\rm tr}(O)^2={\rm tr}(O^{\otimes 2})$, where the first trace is in the Hilbert space on $\nbos$ bosons and the second is in the Hilbert space on $2\nbos$ bosons.  All other contractions can be expressed as ${\rm tr}({\rm SWAP} \, O^{\otimes 2})$, where ${\rm SWAP}$ is some operator swapping bosons in some given permutation and which acts on $2\nbos$ bosons.
Since ${\rm SWAP}$ is unitary and hence has operator norm one, this trace is bounded by the trace norm of $O^{\otimes 2}$, which, since $O$ is positive semi-definite, is equal to ${\rm tr}(O^{\otimes 2}$.  So,
$$\expec[\langle( \Tminus)^{\otimes \nbos(j)/4}| O |(\Tminus)^{\otimes \nbos(j)/4} \rangle^2] \leq 
(2 \nbos(j)/4)! \expec[\langle (\Tminus)^{\otimes \nbos(j)/4}| O |(\Tminus)^{\otimes \nbos(j)/4} \rangle]^2.$$
So, for appropriate choice of the $O$, we can then bound $\expec[P(j)^2] \leq (2 \nbos(j)/4)! \expec[P(j)]^2$ and similarly for $Q(j)$.
So, with high probability $P(j)$ and $Q(j)$ are only $\tilde O(1)$ larger than their mean.
If there is some unlikely event that can lead to $P(j),Q(j)$ being much smaller than their mean, this will only speedup the algorithm as it will accelerate later steps as we will see.

We are choosing every step of the amplitude amplification, except the last (i.e., all $j>0$), so that it would succeed with probability close to $1$ in the unspiked case.  One expects that the success probability in the spiked case is roughly the same on these steps (one can estimate this success probability by how the spiked state differs from the unspiked), but if it is the case that it differs significantly in the spiked case, that simply allows detection at an even earlier step and so would allow an even larger speedup.  So we may assume that in both spiked and unspiked cases, the success probability of these projectors on all steps except the last is $\leq \tilde O(1) P(j)$.

In contrast, we choose the last step of amplitude amplification so that, with high probability, it would succeed only in the spiked case.
We can lower bound the success probability of this last step in the spiked case in the same way as we bounded it before.
That is, the state $\psiideal$ survives all the approximate projectors (up to $o(1)$ corrections), and so that gives a lower bound on the average of $P(j)$ in the spiked case.  We can then use the same Wick's theorem and second moment method to bound fluctuations; to do this we will take the vector $v$ to be the image of the input state under approximate projections $\tilde \Pi(0) (\tilde \Pi(1) \otimes \tilde \Pi(1)) \ldots$.
We then choose the success probability on the last step so that in the spiked case it exceeds this and in the unspiked it does not.
This cutoff is not $\normmin$ but may be taken to be $\normmin/\prod_{j=1}^k P(j)^{2^j}$ due to the probability of not succeeding on previous steps.

So, the runtime of the last step is inverse in the square-root of this cutoff, so it is $\normmin^{-1/2}\cdot \tilde O(1)$ times $\prod_{j=1}^k \sqrt{P(j)^{2^j}}$ times the time for the previous number of steps.
From, $P(j) P(j+1)^2 P(j+2)^4 \ldots \leq Q(j)$, we have that $P(j)\leq Q(j)/(P(j+1)^2 P(j+2)^4 \ldots)$.
Consider the dependence of the runtime on $P(1)$: a smaller $P(1)$ reduces the runtime, so for worst case we may take $P(1)=Q(1)/(P(2)^2 P(3)^4 \ldots)$.
Having taken that, we find in turn that a smaller $P(2)$ also reduces the runtime, and so on for all $j$, so we may take in worst case
 $P(j)= Q(j)/(P(j+1)^2 P(j+2)^4 \ldots)$ to bound the runtime.

 This gives an overall run time $\sqrt{\normmin}^{-1/2} \cdot \prod_{j=1}^k \sqrt{Q(j)} \cdot \tilde O(1)$, which gives a cubic speedup for large $k$.
Here note that $Q(j)=P(0)^{1/4^j}$.  The reason that it decays in this way with $j$ is a combination of two effects.  First, the Hilbert space dimension is decreasing do that the dimension of the Hilbert space on a system of length reduced by $1/2^j$ is equal to the full Hilbert space dimension raised to the power $1/2^j$.  Second, the approximate projector is projecting onto an eigenspace of eigenvalue greater than or equal to the largest eigenvalue times $1/2^{j+1}$ as the energy of the ideal state is decreasing more rapidly with $j$ than the largest eigenvalue.

\section{Recovery}
\label{recovery}
We finally consider recovery.  
We perform recovery by measuring the single particle density matrix of the state $\psi$ that survives the approximate projector $\tilde \Pi$.
Following \cite{wein2019kikuchi}, define the correlation of two vectors $x,y$ by
\be
{\rm corr}(x,y)=\frac{\langle x|y\rangle}{|x|\cdot |y|}.
\ee
It is possible to boost a weak correlation to a stronger correlation, as in \cite{wein2019kikuchi} and implicitly in \cite{montanari2014statistical}, using a tensor power algorithm.
In this manner, for $p=4$, so long as some vector $u$ has ${\rm corr}(u,\vsig)=\omega(N^{-1/6})$, we can obtain a vector with correlation $1-o(1)$ with $\vsig$ with high probability, which is termed strong recovery.
Using the randomized recovery algorithm, Algorithm 1 of \cite{hastings2020classical}, we can recover a $u$ with such a given correlation $\omega(N^{-1/6})$ so long as the single particle density matrix $\rho$ has $\langle \vsig | \rho | \vsig \rangle = \omega(N^{-1/3})$.

Let us break the Hamiltonian $H(\Tplus)$ into the sum of two terms:
$$H(\Tplus)=H(\lambda^{+} \vsig^{\otimes 4}) + H(G^+),$$
where the first term is the contribution of the spike vector to $\Tplus$ and the second term is due to the Gaussian noise.  Let us rotate the basis for the $N$-component vectors so that
$(\vsig)_0=\sqrt{N}$ and all remaining components vanish.
Then $H(\lambda^{+} \vsig^{\otimes 4})=(a^\dagger_0 a_0)^2-a^\dagger_0 a_0.$  We will upper bound the expectation value of this on a state $\psi$ with single particle density matrix $\rho$; the worst case is when with probability $\rho_{00}/\nbos$ all bosons are in state $0$ and otherwise none are.  Hence
\be
\langle \psi | H(\lambda^{+} \vsig^{\otimes 4}) | \psi \rangle \leq (\nbos-1) \langle \vsig | \rho | \vsig \rangle.
\ee

Thus, if we cannot achieve strong recovery by boosting a weak correlation of a vector obtained from the single particle density matrix,
we must have
$$\langle \psi | H(\lambda^{+} \vsig^{\otimes 4}) | \psi \rangle=O(\nbos N^{-1/6}).$$

What this means is, since the state $\psi$ survived the approximated projector $\tilde \Pi$, the expectation value of 
$H(G^+)$ in that state must be large enough; indeed, it must be close to the lower cutoff on eigenvalue for $\tilde \Pi$.  

The discussion from here on in this section is heuristic, rather than rigorous, and the goal is simply to give a plausibility argument that recovery will succeed without proving it.
Intuitively, indeed, it seems very surprising if it did not succeed: we have
shown that if we cannot achieve recovery, we have obtained a state whose expectation value for $H(G^+)$  is fairly large.   This is surprising as our main results above have been based on the small number of eigenstates of $H(G^+)$ with such a large eigenvalue, and so we have shown that for an unspiked tensor, the input state is unlikely to have such a large expectation value for $H(G^+)$ after phase estimation.
Here, however, our goal is recovery and so we have two differences if we wish to those this precisely: the input state has a spike and also the tensor used to define the Hamiltonian has a spike.  It would, however, be somewhat surprising if the spike in the tensor defining the Hamiltonian increased the probability of finding a large expectation value for the unspiked $H(G^+)$.  We leave a proof of recovery for future work.

\section{What is the density of states?}
The previous part of the paper all shows a speedup assuming a certain Alice theorem, which we had when the best bound on the largest eigenvalue was from matrix Chernoff bounds or similar.  However, now we have a tighter bound on the largest eigenvalue\cite{kothari2026smooth}.  Will we expect such a speedup to still exist?

First, since the original quartic quantum speedup\cite{hastings2020classical}, similar quartic quantum speedups have been developed for other problems such as \cite{schmidhuber2025quartic1,schmidhuber2025quartic2}.  The bounds used in those paper are all in the style of a matrix Chernoff bound, so relative to the best bound in those papers, we \emph{do} expect to have the needed density of states bound.  Of course, the problem is different in detail, but we conjecture that the techniques in this paper will carry over to those results too.  Of course, the bounds may change but this gives some optimism that either the bounds will be improved or the polynomial speedup will hold.

What about in the original tensor PCA problem?  What do we expect for the density of states there?  Lower bounds in Ref.~\cite{kothari2026smooth} indicate that there is some constant $c>0$ such that the number of eigenvalues at least $c$ times the largest eigenvalue is proportional to the dimension of the Hilbert space, indicating that the scaling assumption on the density of states does not hold.  However, it does not rule out the possibility that for some $c'$, with $c<c'<1$, the number of eigenvalues at least $c'$ times the largest eigenvalue is polynomially smaller than the dimension of the Hilbert space, which would still give some speedup (or increase in signal-to-noise ratio) using the techniques here.  We have performed numerical simulations (unpublished) of the case $\nbos=4,p=4$ at increasing $N$, and indeed we have found that there is some trend that, as $N$ increases, the number of eigenvalues which are at least some large fraction of the largest eigenvalue decreases.  The simulations are still quite inconclusive of the scaling at large $N$ but there at least is the possibility of some practical speedup.

Another case where it would be likely would be variants of the tensor PCA problem.  To motivate this, the intuition we have is that the needed property (density of states becoming polynomially smaller above some constant fraction of the maximum eigenvalue) is mostly likely to happen when $\nbos$ is large.  Decay in the density of states at large number of particles is standard, for example in the SYK model\cite{garcia2017analytical}.  Normally, there would of course be no reason to take $\nbos$ larger than $N$ in tensor PCA, as an $\epsilon$-net of size exponential in $N$ will find the spike vector, while spectral algorithms take time exponential in $\nbos$.  However, we can imagine taking both $p$ and $N$ large and then the most naive $\epsilon$-net methods will take time exponential in $Np$ so it may become useful in this case to take $\nbos$ large in which case the density of states decay is expected.

\bibliography{superquart}

\end{document}